\documentclass[10pt,journal,letterpaper]{IEEEtran}

\makeatletter
\def\@copyrightspace{\relax}
\makeatother

\usepackage{epstopdf}
\usepackage{graphicx}
\usepackage{stmaryrd}
\usepackage{amsmath,amsthm,amssymb}
\usepackage{multirow,url}
\usepackage{color,cite}
\usepackage{algorithm}
\usepackage{algorithmicx}
\usepackage[algo2e]{algorithm2e}
\usepackage{algpseudocode}
\usepackage{setspace}
\usepackage{rotating}

\newtheorem{thm}{Theorem}
\newtheorem{defn}{Definition}

\ifodd 0
\newcommand{\rev}[1]{{\color{blue}#1}} 
\newcommand{\com}[1]{\textbf{\color{red} (COMMENT: #1)}} 
\newcommand{\comg}[1]{\textbf{\color{green} (COMMENT: #1)}}
\newcommand{\response}[1]{\textbf{\color{magenta} (RESPONSE: #1)}} 
\else
\newcommand{\rev}[1]{#1}
\newcommand{\com}[1]{}
\newcommand{\comg}[1]{}
\newcommand{\response}[1]{}
\fi

\begin{document}

\title{Efficient Resource Allocation for On-Demand Mobile-Edge Cloud Computing}

\author{Xu Chen, Wenzhong Li, Sanglu Lu, Zhi Zhou, and Xiaoming Fu \thanks{

Xu Chen and Zhi Zhou are with School of Data and Computer Science, Sun Yat-sen University, Guangzhou, China. Emails: \{chenxu35,zhouzhi9\}@mail.sysu.edu.cn.

Wenzhong Li and Sanglu Lu are with State Key Laboratory for Novel Software Technology, Nanjing University, Nanjing, China. Emails: \{lwz, sanglu\}@nju.edu.cn.

Xiaoming Fu is with Institute of Computer Science, University of Goettingen, Goettingen, Germany. Email: fu@cs.uni-goettingen.de.

}

}


\maketitle
\pagestyle{empty}
\thispagestyle{empty}

\allowdisplaybreaks

\begin{abstract}
Mobile-edge cloud computing is a new paradigm to provide cloud computing capabilities at the edge of pervasive
radio access networks in close proximity to mobile users. Aiming at provisioning flexible on-demand mobile-edge cloud service, in this paper we propose a comprehensive framework consisting of a resource-efficient computation offloading mechanism for users and a joint communication and computation (JCC) resource allocation mechanism for network operator. Specifically, we first study the resource-efficient computation offloading problem for a user, in order to reduce user's resource occupation by determining its optimal communication and computation resource profile with minimum resource occupation and meanwhile satisfying the QoS constraint. We then tackle the critical problem of user admission control for JCC resource allocation, in order to properly select the  set of  users for resource demand satisfaction.  We show the admission control problem is NP-hard, and hence develop an efficient approximation solution of a low complexity by carefully designing the user ranking criteria and rigourously derive its performance guarantee.  To prevent the manipulation that some users may untruthfully report their valuations in acquiring mobile-edge cloud service, we further resort to the powerful tool of critical value approach to design truthful pricing scheme for JCC resource allocation. Extensive performance evaluation demonstrates that the proposed schemes can achieve superior performance for on-demand mobile-edge cloud computing.
\end{abstract}

\section{Introduction}

As smartphones are gaining enormous popularity, more and more new
mobile applications such as face recognition, natural language processing,
interactive gaming, and augmented reality are emerging and attract
great attention \cite{kumar2010cloud,soyata2012cloud,chen2015exploiting}. This kind of mobile applications
are typically resource-hungry,  demanding intensive computation and real-time responsiveness. Due to the physical size constraint, however,
mobile devices are in general resource-constrained, having limited
computation resources. The tension between
resource-hungry applications and resource-constrained mobile devices
hence poses a significant challenge for the future mobile platform development.


Mobile cloud computing is envisioned as a promising
approach to address such a challenge. By offloading the
computation via wireless access to the resource-rich cloud
infrastructure, mobile cloud computing can augment the capabilities
of mobile devices for resource-hungry applications. Currently, one common approach for mobile cloud computing is to offload the computation-intensive tasks to the remote public cloud infrastructure(e.g., Amazon EC2 and Windows Azure), in order to utilize the powerful computing and processing capabilities by the public clouds. As a matter of fact, the current public cloud architecture - built around static Internet-based installments of cloud resources not integrated with wireless networks - is already starting to show its limits in terms of computation-intensive mobile application support, since mobile users would experience
long latency for data exchange with the public cloud through the wide area network (WAN), which risks to become the major impediment in satisfying the real-time interactive response requirement of mobile applications.


\begin{figure}[t]
\centering
\includegraphics[scale=0.65]{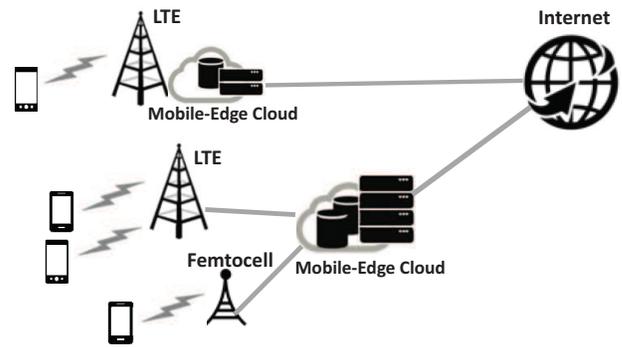}
\caption{\label{fig:An-illustration-of}An illustration of mobile-edge cloud computing }
\end{figure}

To address this challenge, a novel mobile cloud computing paradigm, called mobile-edge cloud computing, has been proposed \cite{MEC2014}. As illustrated in Figure \ref{fig:An-illustration-of}, mobile-edge cloud computing can provide cloud-computing capabilities at the edge of pervasive radio access networks (e.g., 3G, 4G, WiMax, femtocells) in close proximity to mobile users. In this case, the need for fast interactive response can be met by fast  and  low-latency connection (e.g., via fiber transmission) to resource-rich cloud computing infrastructures (called mobile-edge clouds) deployed by network operators (e.g., AT\&T and T-Mobile) within the network edge and backhaul/core networks. By endowing ubiquitous radio access networks with powerful computing capabilities, mobile-edge cloud computing is envisioned to provide pervasive and agile computation augmenting services for mobile users at anytime and anywhere \cite{MEC2014}.


In this paper, we aim at devising an efficient  mobile-edge cloud computing framework that can provide rich flexibility in meeting different mobile users' demands. To this end, in this paper we propose a comprehensive framework consisting of a resource-efficient computation offloading mechanism for the users and a joint communication and computation (JCC) resource allocation mechanism for the network operator. Specifically, we first address the resource-efficient computation offloading problem, in order to reduce a user's  resource occupancy. We then study the admission control problem and design the pricing scheme for JCC resource allocation, by jointly taking into account the objective of system-wide performance optimization as well as the practical constraints such as computational complexity for practical implementation and truthfulness for preventing manipulation.

The main results and contributions of
this paper are as follows:


$\bullet$ \emph{We address the resource-efficient computation offloading problem for each individual user subject to the QoS constraint.} Specifically, by leveraging the structural property of user's task graph, we propose an efficient delay-aware task graph partition algorithm for computation offloading. Building on this, we derive the optimal communication and computation resource demanding profile for a user that minimizes the resource occupancy and meanwhile satisfies the completion time constraint.

$\bullet$ \emph{We study the admission control problem for selecting the proper set of users in JCC resource demand satisfaction.} We show the admission control problem involving joint communication and computation resource allocation is NP-hard. We hence develop an efficient approximation solution by carefully designing the user ranking criteria, which has a low complexity to facilitate the practical implementation.  We further derive the upper bound of the performance loss of the approximate admission control solution with respect to the global optimal solution. Numerical results demonstrate that the proposed approximate solution is very efficient, with at most $15\%$ performance loss.

$\bullet$ \emph{We investigate the truthful pricing problem for preventing manipulation by untruthful valuation reporting.} We borrow the powerful tool of critical value approach in auction theory, and rigorously show that the proposed pricing scheme is truthful such that no user has the incentive to lie about its valuation for acquiring the mobile-edge cloud service to complete its computational task. This would be very useful to  prevent manipulations by untruthful valuation report, which would lead to inefficient allocation of resources and system performance degradation.

The rest of the paper is organized as follows. We first
introduce the system model
in Section \ref{system}. We then study resource-efficient computation offloading strategy for the users in Section \ref{RECO}.  We next discuss the admission control problem and pricing issue for JCC resource allocation in Sections \ref{allocation} and \ref{pricing}, respectively. We conduct the performance evaluation in Section \ref{sec:Numerical-Results}, discuss the related works in Section \ref{sec:Related-Work}, and finally conclude in Section \ref{sec:Conclusion}.

\section{System Model}\label{system}

We consider that there are a set of $K$ wireless base-stations (e.g., macrocell/femtocell
base-stations) $\mathcal{K}=\{1,2,...,K\}$ and a set of $B$ mobile-edge
clouds $\mathcal{L}=\{1,2,...,L\}$ in the system. For a base-station $k\in\mathcal{K}$,
there are $M_{k}$ orthogonal wireless subchannels dedicated for supporting
on-demand mobile-edge cloud service, and each subchannel has a bandwidth
of $w$. Moreover, a base-station $k$ is directly connected to the closest mobile-edge
cloud $l\in\mathcal{L}$ in proximity, which has a total computation
resource capacity of $B_{l}$. Note that as illustrated in Figure \ref{fig:An-illustration-of}, our model allows that multiple nearby base-stations share the same mobile-edge cloud in proximity.

A set of $N$ mobile users $\mathcal{N}=\{1,2,...,N\}$
would like to acquire the mobile-edge cloud service from the network
operator to complete their computation-intensive
tasks. In the following we denote the base-station and the mobile-edge
cloud with which a user $n$ is associated as $k(n)$ and $l(n)$,
respectively. Subject to the QoS requirement of the application, a user
$n\in\mathcal{N}$ also has a maximum allowable completion time $\mathcal{T}_{n}$
for its task, and possesses a valuation $v_{n}$, i.e., the maximum amount of monetary cost
that the user is willing to pay for acquiring the mobile-edge
cloud service to complete its task.


In the following parts, we will first study the resource-efficient computation offloading problem
for each user, in order to reduce user's resource occupation by determining its optimal communication and computation resource profile with minimum resource occupation, and meanwhile satisfying the completion time constraint. We will then design the efficient admission control mechanism for the network operator in order to select the proper set of users for JCC resource allocation subject to the resource capacity constraints. We will further develop a pricing scheme to prevent the manipulation such that some users may untruthfully report their valuations in acquiring mobile-edge cloud service, which would lead to inefficient allocation of resources and system performance degradation.

Note that for ease of exposition, in this paper we consider two most critical resource types for mobile applications, i.e., communication and computation resources. Our model can be easily extended to the cases with more resource types such as storage resource. Also, similar to many existing studies in  wireless resource auction (e.g., \cite{wen2012energy,yang2013framework,ChenToN2015}), as an initial thrust and to enable tractable analysis,  in this paper we consider a static setting such that users are stationary. The more general case that users may dynamically depart and leave the mobile-edge cloud system is very challenging and will be addressed in a future work.


\section{Resource-Efficient Computation Offloading}\label{RECO}

In this section, we consider the resource-efficient computation offloading
problem of each individual user to determine the optimal joint communication and computation
resource profile of a user for the cost-effective auction bidding, in order to minimize the resource occupancy
and meanwhile satisfy the required completion time constraint.

\subsection{Task Graph Model}

To proceed, we first introduce the task graph model to describe the processing procedure of the computational task of a user $n$. Similar to
many studies (e.g., \cite{soyata2012cloud,zhang2012offload,yang2013framework,kao2014optimizing,barbarossa2013joint,chen2018thriftyedge} and references therein), we describe the computation task profile
of a user $n$ by a directed acyclic graph $\mathcal{G}_{n}(\mathcal{V}_{n},\mathcal{E}_{n})$, with the end node as the output component to be executed
at the user's device. Figure \ref{TaskGraph} illustrates an example of the task graph of QR-code recognition application \cite{yang2013framework}. Specifically, in the task graph $\mathcal{G}_{n}(\mathcal{V}_{n},\mathcal{E}_{n})$,
the nodes represent the task components and the directed edges stand
for the data dependencies. For a node $i\in\mathcal{V}_{n}$ in the
task graph $\mathcal{G}_{n}(\mathcal{V}_{n},\mathcal{E}_{n})$,
we denote $C_{i}$ as the total computational instructions (i.e.,
CPU cycles) to execute the corresponding task component. For a subsequential
node $j\in\mathcal{V}_{n}$ that invokes the output data from node
$i$ (i.e., $(i,j)\in\mathcal{E}_{n}$), we denote $D_{ij}$ as the
amount of data that input from node $i$ to node $j$. Note that in practice many applications can be presented
by such a directed acyclic task graph model\footnote{We will consider the cyclic task graph case in a future work. One possible solution is to transform this case to the acyclic case by grouping the nodes in a cycle as a ``supper" node.} \cite{soyata2012cloud,zhang2012offload,yang2013framework,kao2014optimizing,barbarossa2013joint,chen2018thriftyedge}, e.g., face recognition, natural language processing, video streaming, virtual argumentation, data processing applications such as MapReduce.

\begin{figure}[t]
\centering
\includegraphics[scale=0.7]{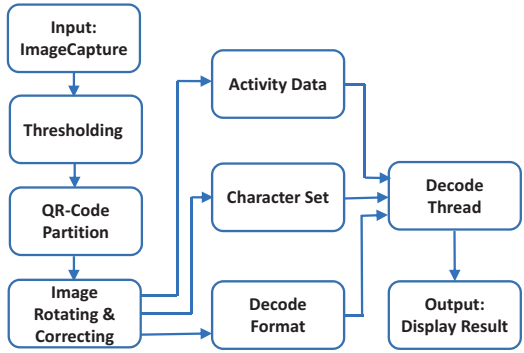}
\caption{\label{TaskGraph}Task graph of QR-code recognition application \cite{yang2013framework}}
\end{figure}

\subsection{Computation and Communication Models}

We then introduce both the communication and computation models. For
a given task component $i$ of user $n$ on its task graph (i.e., $i\in\mathcal{V}_{n}$),
the user can execute it either locally at its own device or remotely
at the mobile-edge cloud. For the local computing approach, we have
the time for executing the given task component $i$ as $T_{i,n}^{m}=\frac{C_{i}}{F_{n}^{0}}$,  where $F_{n}^{0}$ denotes the computation capability (CPU cycles
per second) of user $n$'s device. For the cloud computing approach, a virtual machine is associated to user $n$ for its task execution in the mobile-edge cloud. To provide flexibility in the mobile-edge cloud service provision, we consider
that there is a set of $S$ different virtual machine types $\mathcal{S}=\{1,2,...,S\}$
available and the computation capability of a virtual machine type
$s\in\mathcal{S}$ is denoted by $F_{s}^{c}$. Let $s_{n}$ be the
virtual machine type chosen by user $n$ for executing its whole task. Accordingly, we can then compute the
time by the cloud computing approach for executing the given task component $i$ on the virtual machine as $T_{i,n}^{c}=\frac{C_{i}}{F_{s_{n}}^{c}}$.

Now suppose that two task components $i,j\in\mathcal{V}_{n}$ with
$(i,j)\in\mathcal{E}_{n}$ are executed at different locations (one
at user's device and the other at the mobile-edge cloud). Then a total
amount $D_{ij}$ of data needs to be transferred between these two
components via wireless connection. We hence introduce the communication model for wireless
data transmission. Let $q_{n}\leq M_{k(n)}$ denote the number of
subchannels demanded by user $n$. We can then compute the data rate
for computation offloading of user $n$ as $R_{n}=q_{n}w\log_{2}\left(1+\frac{\eta_{n}h_{n,k(n)}}{\varpi_{n}}\right)$, where $\eta_{n}$ is the transmission power of user $n$ which is determined
by the base station according to some power control algorithm such
as \cite{saraydar2002efficient,barbarossa2013joint}, $h_{n,k(n)}$ is the channel gain between the user $n$
and the base station, and $\varpi_{n}$ is the background noise. Based
on the communication model above, the total data transmission time for
computation offloading between task components $i$ and $j$ is given
as $T_{ij,n}^{tx}=\frac{D_{ij}}{R_{n}}$.

\begin{algorithm}[tt]

\textbf{set} the sorted node list $\widetilde{\mathcal{V}_{n}}=\emptyset$

\textbf{set} $\mathcal{S}$ as the set of all nodes with no incoming edges

\While{$\mathcal{S}\neq \emptyset$}{
   \textbf{remove} a node $i\in \mathcal{S}$

   \textbf{add} node $i$ to the end of the list $\widetilde{\mathcal{V}_{n}}$

   \For{each node $j$ such that $(i,j)\in\mathcal{E}_{n}$}{
        \textbf{remove} edge $(i,j)$ from the task graph $\mathcal{G}_{n}(\mathcal{V}_{n},\mathcal{E}_{n})$
        \If{node $j$ has no other incoming edges}{
             \textbf{add} node $i$ to the set $\mathcal{S}$
        }
   }
}

\textbf{return} the list $\widetilde{\mathcal{V}_{n}}$
\caption{\label{alg:TS}Topology Sorting of Task Graph}
\end{algorithm}

\subsection{Delay-Aware Task Graph Partition\label{sub:2.3}}

We next consider the delay-aware task graph partition problem for a fixed resource allocation. That is, given that a user $n$'s subchannel demand
$q_{n}$ for wireless communication and the virtual machine type selection $s_{n}$
for cloud computation are fixed, we would like to determine the optimal execution
location (i.e., device or cloud) for each task component in its task graph, in order
to minimize the total delay of executing the whole task.

To address the delay-aware task graph partition problem, we leverage
the structural property of the task graph $\mathcal{G}_{n}(\mathcal{V}_{n},\mathcal{E}_{n})$.
Since the graph $\mathcal{G}_{n}(\mathcal{V}_{n},\mathcal{E}_{n})$
is directed acyclic, it is known that we can carry out a topology
sorting of the graph, i.e., ordering of the nodes such that for every
directed edge $(i,j)\in\mathcal{E}_{n}$ from node $i$ to node $j$,
$i$ comes before $j$ in the ordering \cite{cormen2009introduction}.

We introduce the classic Kahn Algorithm \cite{cormen2009introduction} for the topology
sorting of the task graph in Algorithm \ref{alg:TS}, which has a computational complexity of $\mathcal{O}(|\mathcal{V}_{n}|+|\mathcal{E}_{n}|)$
(please refer to \cite{cormen2009introduction} for the details). After the topology sorting,
we then obtain an ordering of the nodes as $\widetilde{\mathcal{V}_{n}}=(a_{1},a_{2},...,a_{|\mathcal{V}_{n}|})$
such that $a_{1}<a_{2}<...<a_{|\mathcal{V}_{n}|}$. We also denote
the set of nodes to which there is an edge directed from node $i$
as $\Delta_{n}(i)=\{j:(i,j)\in\mathcal{E}_{n}\}$. According
to the definition of topology sorting, we know that if $j\in\Delta_{n}(i)$,
then $i<j$ in the ordering $\widetilde{\mathcal{V}_{n}}$. Note that
since no edge is directed from the end node (i.e., the output component)
and for any other node there always exists some edge directed from
it (otherwise the corresponding task component is useless since it does
not provide any input to either other intermediate components or the
output component), the last node $a_{|\mathcal{V}_{n}|}$ in the ordering
$\widetilde{\mathcal{V}_{n}}$ must be the output component.

Based on the node ordering by topology sorting, we then solve the
delay-aware task graph partition problem by the principle of backward
induction \cite{cormen2009introduction}. That is, we first determine the optimal execution location
for the last node in the ordering $\widetilde{\mathcal{V}_{n}}$,
and based on which we can determine the optimal location for the nodes
by moving backward. Specifically, we denote the execution location
of the task component $i$ as $y_{i,n}\in\{m,c\}$, with $y_{i,n}=c$
denoting the component will be executed at the cloud and $y_{i,n}=m$
denoting the component will be executed at the device. Moreover, we
denote the minimum delay starting from executing the task component
$i$ till the end of the whole task as $Z_{n}(i)$.


\begin{algorithm}[tt]

\textbf{conduct} the topology sorting of $\mathcal{G}_{n}(\mathcal{V}_{n},\mathcal{E}_{n})$ and obtain the sorted node list $\widetilde{\mathcal{V}_{n}}$

\textbf{set} $Z_{n}(a_{|\mathcal{V}_{n}|})=T_{a_{|\mathcal{V}_{n}|},n}^{m}$ and $y_{a_{|\mathcal{V}_{n}|},n}^{*}=m$ for the end node $a_{|\mathcal{V}_{n}|}$

\For{$i=|\mathcal{V}_{n}|-1,...,1$}{
   \textbf{compute} $Z_{n}(a_{i})$ and $y_{a_{i},n}^{*}$ according to (\ref{eq:OF1}) and (\ref{eq:OF2}), respectively
}
\textbf{return} $\{Z_{n}(a_{i})\}_{a_{i}\in\mathcal{V}_{n}}$ and $\{y_{a_{i},n}^{*}\}_{a_{i}\in\mathcal{V}_{n}}$
\caption{\label{alg:Delay-Optimal-Task-Graph}Delay-Aware Task Graph Partition}
\end{algorithm}

For the last node $a_{|\mathcal{V}_{n}|}$ (i.e., the output component)
in the ordering $\widetilde{\mathcal{V}_{n}}$, we have $Z_{n}(a_{|\mathcal{V}_{n}|})=T_{a_{|\mathcal{V}_{n}|},n}^{m}$
and the optimal execution location is $y_{a_{|\mathcal{V}_{n}|},n}^{*}=m$,
due to the fact that the final result will be returned to the user's
device. Then we move backward to consider the remaining nodes $a_{i}<a_{|\mathcal{V}_{n}|}$
in the ordering $\widetilde{\mathcal{V}_{n}}$ in a recursive manner.
According to the definition of $Z_{n}(a_{i})$, we know that the
delay starting from executing the task component $a_{i}$ till the
end of the whole task depends on that of the bottleneck node $j\in\Delta_{n}(a_{i})$
with maximum delay that requires the input data from node $i$. Thus,
we have the following:
\begin{equation}
Z_{n}(a_{i})=\min_{y\in\{m,c\}}\max_{j\in\Delta_{n}(a_{i})}\{T_{a_{i},n}^{y}+T_{a_{i}j,n}^{tx}\boldsymbol{1}_{\{y_{j,n}^{*}\ne y\}}+Z_{n}(j)\},\label{eq:OF1}
\end{equation}
 and
\begin{equation}
y_{a_{i},n}^{*}=\arg\min_{y\in\{m,c\}}\max_{j\in\Delta_{n}(a_{i})}\{T_{a_{i},n}^{y}+T_{a_{i}j,n}^{tx}\boldsymbol{1}_{\{y_{j,n}^{*}\ne y\}}+Z_{n}(j)\},\label{eq:OF2}
\end{equation}
where $\boldsymbol{1}_{\{A\}}$ is an indicator function such that
$\boldsymbol{1}_{\{A\}}=1$ if the event $A$ is true and $\boldsymbol{1}_{\{A\}}=0$
otherwise. Here $\boldsymbol{1}_{\{y_{j,n}^{*}\ne y\}}=1$ implies
the wireless data transmission is required since components $i$ and
$j$ are executed at different locations.

By the back induction, for all the remaining nodes $a_{i}<a_{|\mathcal{V}_{n}|}$, we can recursively compute the values of $Z_{n}(a_{i})$ and the optimal execution locations $y_{a_{i},n}^{*}$ according to (\ref{eq:OF1}) and (\ref{eq:OF2}), respectively. We describe the procedure for the delay-aware task graph partition
in Algorithm \ref{alg:Delay-Optimal-Task-Graph}. The complexity of
the backward induction is $\mathcal{O}(|\mathcal{V}_{n}|)$ and the
complexity of topology sorting is $\mathcal{O}(|\mathcal{V}_{n}|+|\mathcal{E}_{n}|)$.
We hence know then complexity of the delay-aware task graph partition
algorithm is at most $\mathcal{O}(|\mathcal{V}_{n}|^{2})$ since $|\mathcal{E}_{n}|\le|\mathcal{V}_{n}|^{2}$.

\subsection{Resource-Efficient Demanding Strategy\label{sub:Resource-Efficient-Computation-O}}

In Section \ref{sub:2.3}, we have solved the delay-aware task graph partition
problem given the fixed subchannel demand $q_{n}$ and virtual machine type selection $s_{n}$ of a user
$n$. Building on this, we then determine the resource-efficient demanding strategy (i.e., the optimal subchannel
demand $q_{n}$ and virtual machine type selection $s_{n}$) of user $n$ to minimize
the resource occupancy, while guaranteeing its task can be completed
within the completion time deadline $\mathcal{T}_{n}$.

First of all, we define the resource
occupancy function for a user $n$ as the sum of the normalized occupancy
ratios of the user for both communication and computation resources,
i.e.,
\begin{equation}
\Phi_{n}(q_{n},s_{n})=\frac{q_{n}}{M_{k(n)}}+\frac{F_{s_{n}}^{c}}{B_{l(n)}}.\label{ROF}
\end{equation}
Recall that $M_{k(n)}$ and $B_{l(n)}$ are the total number of available
subchannels and computation capacity at user n's associated base-station
and mobile-edge cloud, respectively. The key motivation of minimizing the resource occupancy is two-fold. On one hand, in the operator's admission control mechanism design for later,
the resource occupancy serves as an important indicator for accepting
a user's request and determining its payment. Thus, from the user's point of view, a smaller
resource occupancy implies a higher chance of getting admitted into
the mobile-edge cloud service and a lower service payment as well.  On the other hand, from
the perspective of the network operator, it helps to increase the resource
utilization efficiency and support more users in the mobile-edge cloud
service.

Then, the resource-efficient demanding strategy can be obtained by solving the following optimization problem:
\begin{eqnarray}
 & \min_{q_{n},s_{n}} & \Phi_{n}(q_{n},s_{n}) \label{obj1} \\
 & \mbox{s.t.} & Z_{n}(i)\leq\mathcal{T}_{n},\forall i\in\mathcal{V}_{n}.\label{constraint1}
\end{eqnarray}
The objective is to minimize the resource occupancy, and the constraint
guarantees that the whole task (i.e., all the task components) will
be completed within the deadline $\mathcal{T}_{n}$. \rev{Building upon
the delay-aware task graph partition algorithm in Algorithm \ref{alg:Delay-Optimal-Task-Graph},
this problem can be solved as follows: first, we rank all the communication and computation resource profiles $(q_{n},s_{n})$ of a user according to the corresponding resource
occupancy function value $\Phi_{n}(q_{n},s_{n})$ (with ties randomly broken). Then, according to the ranking, we sequentially select a resource profile $(q_{n},s_{n})$ and utilize Algorithm \ref{alg:Delay-Optimal-Task-Graph} to compute the
value of $Z_{n}(i)$ for any $i\in\mathcal{V}_{n}$. Once the constraint in (\ref{constraint1}) is satisfied, then we stop and find
the optimal resource profile $(q^{*}_{n},s^{*}_{n})$ having the minimum value of $\Phi_{n}(q_{n},s_{n})$. The key idea is that, by the backward induction  based on (\ref{eq:OF1}) and (\ref{eq:OF2}), Algorithm \ref{alg:Delay-Optimal-Task-Graph} finds the optimal policy of the task graph partition that minimizes the total executing delay for given a fixed communication and computation resource allocation profile $(q_{n},s_{n})$. So, the optimal resource allocation profile for the delay constrained optimization problem in (\ref{obj1}) can be obtained by exhaustive searching over all the feasible resource allocation profiles. To do it in an efficient way, here we first rank the objective function value $\Phi_{n}(q_{n},s_{n})$ (\ref{ROF}) for all the feasible resource allocation profiles in an increasing manner, and then sequentially search for the optimal resource profile with the minimum value $\Phi_{n}(q_{n},s_{n})$ and meanwhile the delay constraint in (\ref{constraint1}) being satisfied.}

Note that the resource-efficient demanding strategy can be locally computed by each user based on its individual application information. Since there are at most $M_{k(n)}$ subchannels and $S$ different
virtual machine types, in the worst case a user $n$ only has to check at
most $M_{k(n)}S$ resource profiles, then the total complexity of computing the resource-efficient demanding strategy in (\ref{obj1}) is at most $\mathcal{O}(M_{k(n)}S|\mathcal{V}_{n}|^{2})$, which is computationally feasible.


\section{Admission Control for JCC Resource Allocation}\label{allocation}
\rev{After receiving users' resource demanding profiles, a key challenge would be that the total available communication and computation resources are limited in practice and the network operator may not satisfy all resource demands of the users. We hence have to design a user admission control mechanism for the network operator to select the proper set of requesting users to serve, in order to optimize the system-wide performance subject to both communication and computation capacity constraints.}

\subsection{Problem Formulation}

We first formally state the admission control problem for JCC resource
allocation as follows.
\begin{defn}
(Admission Control Problem) Given the resource demanding profiles of all the users,
the operator would like to solve the following admission control
problem:
\begin{eqnarray}
 & \max_{\mathcal{C}\subseteq\mathcal{N}} & \sum_{n\in\mathcal{C}}v_{n}\label{eq:WDP1}\\
 & \mbox{s.t.} & \sum_{n:n\in\mathcal{C},k(n)=i}q_{n}^{*}\leq M_{i},\forall i\in\mathcal{K},\label{eq:WDP2}\\
 &  & \sum_{n:n\in\mathcal{C},l(n)=j}F_{s_{n}^{*}}^{c}\leq B_{j},\forall j\in\mathcal{L}.\label{eq:WDP3}
\end{eqnarray}
\end{defn}
\rev{The objective in (\ref{eq:WDP1}) of the admission control problem
is to select the proper set of users $\mathcal{C}$ to serve, in order to maximize the social welfare (i.e., system-wide performance) which is the sum of the accepted users' valuations.} In other words, we would like to efficiently allocate the resources to the set of users that value those resources most. The first constraint in (\ref{eq:WDP2})
means that for a base-station $i$, the total requested communication
resources (i.e., subchannels) by the selected users should not exceed its capacity. The second constraint in (\ref{eq:WDP3}) represents
that for a mobile-edge cloud $j$, the total requested computation
resources by the selected users should not exceed its total available computation
capacity. In this paper we aim at designing efficient resource allocation scheme for optimizing the system-wide performance (i.e., social welfare maximization), and the
case of maximizing the revenue of the operator will be considered in a future work. Note that since the valuation of a user characterizes its willingness to pay, maximizing the sum of the valuations would also have a positive effect on increasing
the revenue obtained by the operator.

\begin{defn}
(Computational Efficiency) A mechanism is computationally
efficient if the computing procedure of the mechanism terminates in
polynomial time.
\end{defn}
The computational efficiency of a devised algorithm for JCC resource allocation is very important
in practice. Any optimal algorithm with high complexity would require a large amount of time overhead to compute the optimal solution, which leads to high latency and slow responsiveness and hence is useless in reality.

\subsection{Computational Complexity of Admission Control Problem }

We first consider the computational complexity of the admission control problem. As mentioned above,
to make the JCC resource allocation useful in practice, we need to develop
a computationally efficient algorithm for the admission control problem.
Unfortunately, we show that it is NP-hard to compute an optimal solution
for the admission control problem.
\begin{thm}\label{thm:NP}
The admission control problem for the JCC resource allocation is NP-hard.\end{thm}
\begin{proof}
We consider the special case with only one base station having $M$
subchannels available and one mobile-edge cloud having $B$ computation
capacity. In this case, the admission control problem becomes

\begin{eqnarray}
 & \max_{\mathcal{C}\subseteq\mathcal{N}} & \sum_{n\in\mathcal{C}}v_{n}\label{eq:WDP1-1}\\
 & \mbox{s.t.} & \sum_{n:n\in\mathcal{C}}q_{n}^{*}\leq M,\label{eq:WDP2-1}\\
 &  & \sum_{n:n\in\mathcal{C}}F_{s_{n}^{*}}^{c}\leq B.\label{eq:WDP3-1}
\end{eqnarray}
Actually, the problem above is a two-dimensional knapsack problem,
by regarding the communication and computation resources as two deimensions.
It is known that the two-dimensional knapsack problem is strongly
NP-hard \cite{caprara2004two}. Thus, as a generalized extension of the two-dimensional
knapsack problem, the admission control problem in our case is also NP-hard.
\end{proof}

The NP-hardness of the admission control problem implies that we have to develop approximation algorithms for designing a fast admission control mechanism.

\subsection{Approximate Admission Control Algorithm}
We then propose an approximate approach to solve the admission control
problem. The key idea is to first
design the proper ranking metric to order the demanding profiles of all the users,
and then select the set of accepted users in a sequential manner.

Specifically, given the set of demanding profiles by all users $\{\varphi_{n}=<(q_{n}^{*},s_{n}^{*}),v_{n}>\}_{n\in\mathcal{N}}$,
the operator first orders the bids according to the following performance
metric (with tie randomly broken):
\begin{equation}
\gamma_{n}=\frac{v_{n}}{\Phi_{n}(q_{n}^{*},s_{n}^{*})}.\label{eq:pr}
\end{equation}
Recall that $v_{n}$ is user's valuation and $\Phi_{n}(q_{n}^{*},s_{n}^{*})$ is the resource occupancy function in (\ref{ROF}).
The physical meaning of
the performance metric $\gamma_{n}$ indicates the user's valuation
per unit resource occupancy. Intuitively, if a user has
a higher valuation and consumes fewer resources, then the user should
have a higher chance to get accepted in the JCC resource allocation.

After ranking all users, the operator then sequentially selects the set
of accepted users one by one. For a candidate user $n$,
the operator checks, by selecting this user as an accepted user, whether
the total requested communication and computation resources exceed the capacity of its associated base-station and the mobile-edge cloud.
If the capacity constraint is not violated, then the operator will accept this user as an accepted user; otherwise,
the operator will reject this user. We summarize the whole procedure
for determining the accepted users in Algorithm \ref{alg:Winner-Determination-Algorithm}.

\begin{algorithm}[tt]

\textbf{initialize} $x_{1}=...=x_{N}=0$

\For{$n\in\mathcal{N}$}{
    \textbf{compute} $\gamma_{n}=\frac{\lambda_{n}}{\Phi_{n}(q_{n}^{*},s_{n}^{*})}$
}

\textbf{sort} all the users in $\mathcal{N}$ in the decreasing order of $\gamma_{n}$

\For{$n=1,2,...,N$ in the sorted set $\mathcal{N}$}{
    \If{$\sum_{i<n:k(i)=k(n)}q_{i}^{*}x_{i}+q_{n}^{*}\leq M_{k(n)}$ and $\sum_{i<n:k(i)=k(n)}F_{s_{i}^{*}}^{c}x_{i}+F_{s_{n}^{*}}^{c}\leq B_{l(n)}$}{
        \textbf{set} $x_{n}=1$
    }
}
\textbf{return} $\{x_{n}\}_{n\in\mathcal{N}}$
\caption{\label{alg:Winner-Determination-Algorithm} Admission Control for JCC Resource Allocation}
\end{algorithm}

%

\subsection{Performance Analysis}

For the approximate admission control in Algorithm \ref{alg:Winner-Determination-Algorithm},
it has the complexity of $\mathcal{O}(N)$ to compute the ranking
metrics of all users, $\mathcal{O}(N\log N)$ to sort the users, $\mathcal{O}(N)$
to select the accepted users. Thus the complexity of Algorithm \ref{alg:Winner-Determination-Algorithm}
is $\mathcal{O}(N\log N)$. Hence, the approximate admission control is computationally efficient for JCC resource allocation.

We next analyze the worst-case performance of the proposed admission control
algorithm, with respect to the optimal solution that maximizes the
social welfare in (\ref{eq:WDP1}). Let $\mathcal{C}$ and $\mathcal{C}^{*}$
denote the set of accepted users by the admission control algorithm
in Algorithm \ref{alg:Winner-Determination-Algorithm} and the optimal
solution in (\ref{eq:WDP1}), respectively. We also denote the social welfare
by the admission control algorithm and the optimal solution as
$W=\sum_{n\in\mathcal{C}}v_{n}$ and $W^{*}=\sum_{n\in\mathcal{C}^{*}}v_{n}$,
respectively. We define the approximation ratio $0\leq\rho\leq1$
of the admission control algorithm such that $W\geq\rho W^{*}$.

We first consider the special case that the total communication resource reserved for mobile-edge cloud service at each base station is homogeneous and each mobile user demands the same communication resource amount, i.e., $M_{k}=M$ for any $k\in\mathcal{K}$ and $q^{*}_{n}=q$ for any $n\in\mathcal{N}$. This could correspond to the situation that all the base-stations reserve the same communication resource for mobile-edge could computing service and each user is allowed to utilize one channel for computation offloading in practice. For this special case, we have the following result.
\begin{thm}\label{thm:AR2}
For the case that $M_{k}=M$ for any base-station $k\in\mathcal{K}$ and $q^{*}_{n}=q$ for any user $n\in\mathcal{N}$, the admission control algorithm in JCC resource allocation can approximate the optimal solution with a ratio of $\rho=\frac{1}{2}$. \end{thm}
\begin{proof}
For this special case, the admission control problem can
be reformulated as
\begin{eqnarray*}
 & \max_{x_{n}} & \sum_{n=1}^{N}v_{n}x_{n}\\
 & s.t. & \sum_{n:k(n)=i}x_{n}\le\frac{M}{q},\forall i\in\mathcal{K},\\
 &  & \sum_{n:l(n)=j}F_{s_{n}^{*}}^{c}\leq B_{j},\text{\ensuremath{\forall}j\ensuremath{\in\mathcal{L}}},\\
 &  & x_{n}\in\{0,1\},\forall n\in\mathcal{N}.
\end{eqnarray*}

Since a base-station is connected to only one mobile-edge cloud, in
this case we can decompose the problem above by only considering the
sub-problem for each mobile-edge cloud as follows. Here we denote
the set of base-stations and users connecting to base-station $j$
as $\mathcal{K}_{j}$ and $\mathcal{N}_{j}$, respectively.
\begin{eqnarray}
 & \max_{x_{n}} & \sum_{n\in\mathcal{N}_{j}}v_{n}x_{n}\label{eq:opt333}\\
 & s.t. & \sum_{n:k(n)=i}x_{n}\le\frac{M}{q},\forall i\in\mathcal{K}_{j},\nonumber \\
 &  & \sum_{n:l(n)=j}F_{s_{n}^{*}}^{c}\leq B_{j},\nonumber \\
 &  & x_{n}\in\{0,1\},\forall n\in\mathcal{N}_{j}.\nonumber
\end{eqnarray}

Now suppose we relax the problem above by allowing the variable $x_{n}$
taking fractional value, i.e., $x_{n}\in[0,1]$ as follows.
\begin{eqnarray}
 & \max_{x_{n}} & \sum_{n\in\mathcal{N}_{j}}v_{n}x_{n}\label{eq:opt444}\\
 & s.t. & \sum_{n:k(n)=i}x_{n}\le\frac{M}{q},\forall i\in\mathcal{K}_{j},\label{eq:c111}\\
 &  & \sum_{n:l(n)=j}F_{s_{n}^{*}}^{c}x_{n}\leq B_{j},\label{eq:c222}\\
 &  & 0\leq x_{n}\leq1,\forall n\in\mathcal{N}_{j}.\nonumber
\end{eqnarray}
For the relaxed problem above, we can easily solve it by the greedy
manner using $r_{n}=\frac{v_{n}}{F_{s_{n}^{*}}^{c}}$ as the ranking
metric. Specifically, we first rank all the users in the set $\mathcal{N}_{j}$
accordingly to the ranking metric $r_{n}=\frac{v_{n}}{F_{s_{n}^{*}}^{c}}$
and then sequentially add a user into the selected set (i.e., $x_{n}=1$)
if both conditions (\ref{eq:c111}) and (\ref{eq:c222}) are satisfied.
When a user violates condition (\ref{eq:c111}), we set $x_{n}=0$.
When a user who first violates (\ref{eq:c222}), we set $x_{n}=\alpha\in[0,1)$
as a fractional value to ensures the condition (\ref{eq:c111}) is
satisfied and the algorithm ends. We can prove that this greedy algorithm
can find the optimal solution to problem (\ref{eq:opt444}) by contradiction.
This is because, if there exists a user $n$ of a higher ranking metric
$r_{n}$ that is selected by the greedy algorithm but not included
in the final solution, we can swap it with another user $m$ of lower
ranking metric $r_{m}$ that is not picked by the greedy algorithm
but is included in the solution to improve the solution. In particular,
if there exists such a user $m$ that belongs to the same base-station
with user $n$, then we swap user $n$ with user $m$ in the solution
and the conditions (\ref{eq:c111}) and (\ref{eq:c222}) are still
satisfied. If not, we we swap user $n$ with user $m$ from another
base-station. Since in the greedy algorithm user $n$ is selected
and the conditions (\ref{eq:c111}) and (\ref{eq:c222}) are satisfied,
then after the swap the the conditions are still satisfied. As a result,
we see that the solution can be improved and this contradicts with
the optimal assumption.

Without loss generality, we denote the optimal solution for problem
(\ref{eq:opt444}) as $x_{1}=...=x_{d}=1$, $x_{d+1}=\text{\ensuremath{\alpha}}$
and $x_{n}=0$ for any $n>d+1$. Since $q_{n}^{*}=q$ and $M_{k}=M$,
then the ranking metric in the admission control algorithm becomes
$r_{n}=\frac{v_{n}}{q/M+F_{s_{n}^{*}}^{c}/B_{j}}$. This is equivalent
to using the ranking metric $r_{n}=\frac{v_{n}}{F_{s_{n}^{*}}^{c}}$
in the admission control algorithm. As a result, the admission control
algorithm can obtain the solution as $x_{1}=...=x_{d}=1$ and $x_{n}=0$
for any $n>d$. Let $W_{F}$ and $W_{G}$ denotes the values by the
two solutions $(1,...,1,\alpha,0,...,0)$ and $(1,...,1,0,0,...,0)$
above, and $W_{OPT}$ denotes the optimal solution to the original
problem (\ref{eq:opt333}). Then, we have
\[
W_{OPT}\le W_{F}\leq W_{G}+\alpha v_{d+1}\leq2W_{G}.
\]
 Thus, we have proved that $W_{G}\geq\frac{1}{2}W_{OPT}$.
\end{proof}

\rev{The key idea of the proof for Theorem (\ref{thm:AR2}) above is that we first formulate the special case admission control problem as an integer programming problem, and then show that the proposed approximate admission control algorithm can obtain the optimal solution for its relaxed fractional programming problem. By exploring the connection between the original integer programming problem and the relaxed fractional programming problem, we can finally show that in the worst case the approximate admission control algorithm can achieve approximation ratio of 1/2 for the original problem.}

For the general case that each base station may have heterogenous total communication resource and each mobile user may demand the different amount of communication resource, we can show the following result.
\begin{thm}\label{thm:AR}
For the general case, the admission control algorithm in JCC resource allocation can approximate
the optimal solution with a ratio of $\rho=\frac{A_{\max}+B_{\max}}{2A_{\max}B_{\max}},$
where $A_{\max}=\max_{k\in\mathcal{K}}\{M_{k}\}$ and $B_{\max}=\max_{l\in\mathcal{L},s\in\mathcal{S}}\{\frac{B_{l}}{F_{s}^{c}}\}$. \end{thm}
\begin{proof}
Based on the ranking metric $\gamma_{n}$, we first define the set
of users $\Omega_{n}$ of a user $n$ such that if $i\in\Omega_{n}$
then $i\geq n$ and $i\in\mathcal{C}^{*}$ but $i\notin\mathcal{C}$
due to that user $n$ blocks user $i$ being accepted in the JCC resource allocation. In the following we also define that $\widetilde{\Omega}_{n}=\Omega_{n}\cup\{n\}$.
According to the admission control algorithm, we know that for
any user $i\in\widetilde{\Omega}_{n}$, we have $\gamma_{i}\leq\gamma_{n}$,
$k(n)=k(i)$ and $l(n)=l(i)$, which implies that
\[
v_{i}\leq v_{n}\frac{\Phi_{i}(q_{i}^{*},s_{i}^{*})}{\Phi_{n}(q_{n}^{*},s_{n}^{*})}=v_{n}\left(\frac{q_{i}^{*}}{M_{k(n)}}+\frac{F_{s_{i}^{*}}^{c}}{B_{l(n)}}\right)/\left(\frac{q_{n}^{*}}{M_{k(n)}}+\frac{F_{s_{n}^{*}}^{c}}{B_{l(n)}}\right).
\]
By summing over all the users $i\in\widetilde{\Omega}_{n}$, we have
\[
\sum_{i\in\widetilde{\Omega}_{n}}v_{i}=v_{n}/\left(\frac{q_{n}^{*}}{M_{k(n)}}+\frac{F_{s_{n}^{*}}^{c}}{B_{l(n)}}\right)\sum_{i\in\widetilde{\Omega}_{n}}\left(\frac{q_{i}^{*}}{M_{k(n)}}+\frac{F_{s_{i}^{*}}^{c}}{B_{l(n)}}\right).
\]
Since $\sum_{i\in\widetilde{\Omega}_{n}}\left(\frac{q_{i}^{*}}{M_{k(n)}}+\frac{F_{s_{i}^{*}}^{c}}{B_{l(n)}}\right)\leq2$
and $\frac{q_{n}^{*}}{M_{k(n)}}+\frac{F_{s_{n}^{*}}^{c}}{B_{l(n)}}\geq\frac{1}{A_{\max}}+\frac{1}{B_{\max}}$,
we then have
\[
v_{n}\geq\frac{A_{\max}+B_{\max}}{2A_{\max}B_{\max}}\sum_{i\in\widetilde{\Omega}_{n}}v_{i}.
\]
Furthermore, we know that $\mathcal{C}^{*}\subseteq\bigcup_{n\in\mathcal{C}}\widetilde{\Omega}_{n}$.
This implies that
\begin{eqnarray*}
\sum_{n\in\mathcal{C}}v_{n} & \geq & \frac{A_{\max}+B_{\max}}{2A_{\max}B_{\max}}\sum_{n\in\mathcal{C}}\sum_{i\in\widetilde{\Omega}_{n}}v_{i}\\
 & \geq & \frac{A_{\max}+B_{\max}}{2A_{\max}B_{\max}}\sum_{i\in\bigcup_{n\in\mathcal{C}}\widetilde{\Omega}_{n}}v_{i}\\
 & \geq & \frac{A_{\max}+B_{\max}}{2A_{\max}B_{\max}}\sum_{i\in\mathcal{C}^{*}}v_{i}.
\end{eqnarray*}
Thus, we have $W\geq\frac{A_{\max}+B_{\max}}{2A_{\max}B_{\max}}W^{*},$
which completes the proof.
\end{proof}

Theorem \ref{thm:AR} describes the performance guarantee of the admission control
algorithm in the worst case (which rarely happens in practice). Numerical results demonstrate that in practice the proposed algorithm is very efficient, with at most $14.3\%$ performance loss, compared
to the optimal solution.

\section{Truthful Pricing for JCC Resource Allocation}\label{pricing}

As discussed above, the valuation $v_{n}$ of a user $n$ measures the
importance or the utility that the user $n$ can achieve for completing
its task within the delay constraint. It plays a critical role in the admission control for JCC resource allocation. As a result, it would risk that some users may untruthfully report their valuations $v_{n}$ in order to increasing their chances in getting accepted in JCC resource allocation. These manipulations by the users
would lead to inefficient allocation of resources and system-wide performance
degradation. Thus, it is highly desirable to prevent such manipulations
by designing truthful pricing scheme such that all users have incentives to report their true valuations.

Formally, if a user $n$ pays a monetary cost $p_{n}$ for its demanded resources, the user receives the net payoff
of $v_{n}-p_{n}$. Then the truthful pricing problem for JCC resource allocation is defined as follows.

\begin{defn}
(Truthful Pricing Problem) For a user $n\in\mathcal{N}$, let $\varphi_{n}=<(q_{n}^{*},s_{n}^{*}),v_{n}>$
 and $\widetilde{\varphi}_{n}=<(q_{n}^{*},s_{n}^{*}),\widetilde{v}_{n}>$ denote the demand report profiles with the truthful valuation $v_{n}$ and untruthful valuation $\widetilde{v}_{n}$, respectively. We also denote the payments of user $i$ under the truthful demand report $\varphi_{n}$
and the untruthful report $\widetilde{\varphi}_{n}$ as $p_{n}(\varphi_{n})$
and $p_{n}(\widetilde{\varphi}_{n})$, respectively. Then net payoffs
of user $n$ for the truthful report and the untruthful report are $v_{n}-p_{n}(\varphi_{n})$
and $v_{n}-p_{n}(\widetilde{\varphi}_{n})$, respectively. The truthful
pricing problem is to design a payment scheme such that $v_{n}-p_{n}(\varphi_{n})\geq v_{n}-p_{n}(\widetilde{\varphi}_{n}),$
i.e., no user can benefit by reporting its valuation untruthfully.
\end{defn}

\subsection{Truthful Pricing Scheme}

We next consider the truthful pricing scheme design. Here we borrow the critical value approach in auction theory \cite{milgrom2004putting} for our design. Let $\lambda_{n}$ denote user $n$'s reported valuation. Intuitively, a critical
value $\theta_{n}$ of a user $n$ is the minimum value that user
$n$ should pay in order to be accepted in admission control for JCC resource allocation. That is, if $\lambda_{n}>\theta_{n}$,
then user $n$ is accepted; if $\lambda_{n}<\theta_{n}$, then user $n$
is rejected. And accordingly the truthful pricing scheme based the critical value
is that if a user $n$ is accepted, then user pay $\theta_{n}$;
otherwise, it pays $0$.

The key challenge of the critical value approach is determining the critical value. For our problem, we show that determining the critical
value $\theta_{n}$ of a user $n$ is equivalent to identifying its critical user, whose demanding profile is the first
profile following user $n$ that has been rejected in admission control but would have been
granted if user $n$ is absent in the JCC resource allocation. Note that if a user
does not have a critical user, it has a critical value of zero.

To proceed, in order to find the critical user of an accepting user $n$, we first
assume user $n$'s demanding profile is not in the ranked user list
by the admission control algorithm. Then we identify the first
user $i$ (following user $n$ in the ranked user list) in the remaining
set that if that user is selected as an accepted user (i.e., satisfy the
resource constraints) it would make some resource constraint violated
when we also add user $n$ in the accepted list. Thus, this user $i$
is indeed the critical user of user $n$, since user $i$ will not
be accepted in admission control when user $n$ has already been accepted. Accordingly, we can then compute the critical value of
the user $n$ as
\begin{equation}
\theta_{n}=\frac{\lambda_{i}\Phi_{n}(q_{n}^{*},s_{n}^{*})}{\Phi_{i}(q_{i}^{*},s_{i}^{*})}.\label{eq:payment}
\end{equation}

The procedure for obtaining the truthful pricing of JCC resource auction based on critical value is summarized
in Algorithm \ref{alg:Winner-Determination-Algorithm-1}. In the following we will show that the proposed payment scheme will guarantee that the proposed pricing scheme is indeed truthful.


\begin{algorithm}[tt]

\textbf{initialize} $p_{1}=...=p_{N}=0$

\For{$n=1,2,...,N$, sorted in the decreasing order of $\gamma_{n}$}{
    \If{$x_{n}=1$}{
        \textbf{compute} $\beta_{1}=\sum_{j<n:k(j)=k(n)}q_{j}^{*}x_{j}$ and $\beta_{2}=\sum_{j<n:l(j)=l(n)}F_{s_{j}^{*}}^{c}x_{j}$

        \For{$i=n+1,...,N$}{
            \If{$k(i)=k(n)$}{
               \If{$\beta_{1}+q_{i}^{*}\leq M_{k(n)}$ and $\beta_{2}+F_{s_{i}^{*}}^{c}\leq B_{l(n)}$}{
                \textbf{compute} $\beta_{1}=\beta_{1}+q_{i}^{*}$ and $\beta_{2}=\beta_{2}+F_{s_{i}^{*}}^{c}$

                \If{either $\beta_{1}+q_{n}^{*}>M_{k(n)}$ or $\beta_{2}+F_{s_{n}^{*}}^{c}>B_{l(n)}$}{
                 \textbf{set} $p_{n}=\theta_{n}=\frac{\lambda_{i}\Phi_{n}(q_{n}^{*},s_{n}^{*})}{\Phi_{i}(q_{i}^{*},s_{i}^{*})}$

                 \textbf{break}
                }
                }
           }
        }
    }
}
\textbf{return} $\{p_{n}\}_{n\in\mathcal{N}}$
\caption{\label{alg:Winner-Determination-Algorithm-1} Pricing Scheme for JCC Resource Allocation}
\end{algorithm}

%
%
%
%
%
%
%
%
%
%
%
%
%

%
%

\subsection{Truthfulness Analysis}

We first show that the proposed pricing scheme for JCC resource allocation is truthful. According to
the critical value approach for auction mechanism design \cite{milgrom2004putting}, a pricing scheme is truthful if the following conditions hold: (1) the winner determination
algorithm (i.e., the admission control algorithm for our problem) is monotone; (2) each winning user (i.e., each accepted user for our problem) is paid its critical
value. To proceed, we formally define these two conditions.
\begin{defn}
(Monotonicity) An admission control algorithm is monotone, if for
a user $n$ when a report $\varphi_{n}=<(q_{n}^{*},s_{n}^{*}),\lambda_{n}>$
is accepted, then given the reports of the other users $\varphi_{-n}$
are fixed, the report $\widetilde{\varphi}_{n}=<(\widetilde{q}_{n},\widetilde{s}_{n}),\widetilde{\lambda}_{n}>$
satisfying $\widetilde{q}_{n}\geq q_{n}^{*}$, $\widetilde{s}_{n}\geq s_{n}^{*}$
and $\widetilde{\lambda}_{n}\geq\lambda_{n}$ by user $n$ would also be
accepted.
\end{defn}
Intuitively, an admission control  algorithm is monotone if, increasing
the valuation and/or decreasing the size of the desired communication
and computation resources in the bid will not cause an accepted user
to lose given the reports of the other users are fixed.
\begin{defn}
(Critical Value) For a user $n$ in the accepted user set, given the reports of the other users
$\varphi_{-n}$ are fixed, there is a critical value $\theta_{n}$
such that if the user $n$ declares a valuation $\lambda_{n}>\theta_{n}$
in its report, then it must be accepted; if the user $n$ declares a valuation
$\lambda_{n}<\theta_{n}$, then it will be rejected.
\end{defn}
We then prove that the proposed pricing scheme is truthful by showing that it satisfies
the two conditions above.
\begin{thm}
The proposed pricing scheme in Algorithm \ref{alg:Winner-Determination-Algorithm-1} for JCC resource allocation is truthful.\end{thm}
\begin{proof}
We first check the condition of monotonicity. Recall that we use the
performance metric $\gamma_{n}=\frac{\lambda_{n}}{\Phi_{n}(q_{n}^{*},s_{n}^{*})}$
to rank the users in the admission control algorithm. Assume that
a user $n$ with a report $\varphi_{n}=<(q_{n}^{*},s_{n}^{*}),\lambda_{n}>$
is an accepted user. Then given the reports of the other users $\varphi_{-n}$
are fixed, for the user $n$ with the report $\widetilde{\varphi}_{n}=<(\widetilde{q}_{n},\widetilde{s}_{n}),\widetilde{\lambda}_{n}>$
satisfying $\widetilde{q}_{n}\geq q_{n}^{*}$, $\widetilde{s}_{n}\geq s_{n}^{*}$
and $\widetilde{\lambda}_{n}\geq\lambda_{n}$, we know that $\widetilde{\gamma}_{n}=\frac{\widetilde{\lambda}_{n}}{\Phi_{n}(\widetilde{q}_{n},\widetilde{s}_{n})}\geq\gamma_{n}=\frac{\lambda_{n}}{\Phi_{n}(q_{n}^{*},s_{n}^{*})}$.
Then the report $\widetilde{\varphi}_{n}$ must also be an accepted report
according to the admission control algorithm. Thus, the condition
of monotonicity is satisfied.

We then check the condition of critical value. Recall that the critical
value of an accepted user $n$ is given as $\theta_{n}=\frac{\lambda_{i}\Phi_{n}(q_{n}^{*},s_{n}^{*})}{\Phi_{i}(q_{i}^{*},s_{i}^{*})}$
where user $i$ is the critical user of user $n$, whose report is the
first report following user $n$ that has been rejected but would have
been granted if user $n$ is absent in accepted user set. Thus, if user
$n$ declares a valuation $\lambda_{n}>\theta_{n}$ in its report, then user $n$ has a higher priority than the critical user $i$ in the admission control since $\gamma_{n}=\frac{\lambda_{n}}{\Phi_{n}(q_{n}^{*},s_{n}^{*})}>\gamma_{i}=\frac{\lambda_{i}}{\Phi_{i}(q_{i}^{*},s_{i}^{*})}$.
In this case the set of accepted users does not change and user $n$ must still be accepted.
If user declares a valuation $\lambda_{n}<\theta_{n}$ in its report,
then we have $\gamma_{n}<\gamma_{i}$, i.e., user $i$ has a lower
priority than the critical user $i$ in the admission control. In this case, according to Algorithm \ref{alg:Winner-Determination-Algorithm-1} for determining the critical user, the user $i$ will be accepted, and if we further add user $n$ in the accepted user set the resource capacity constraint
is violated. As a result, user $n$ can not be accepted. Thus, the condition of critical value also holds. Combining the arguments above, we hence know that the proposed pricing scheme for JCC resource allocation is truthful.
\end{proof}

Hence, by adopting the truthful pricing scheme in Algorithm \ref{alg:Winner-Determination-Algorithm-1}, all users have incentives to report their true valuations to the network operator. This will be very useful to prevent manipulations by untruthful reports and guarantee the system efficiency for mobile-edge cloud computing.

%

\section{Numerical Results}\label{sec:Numerical-Results}

\begin{figure}[tt]
\centering
\includegraphics[scale=0.45]{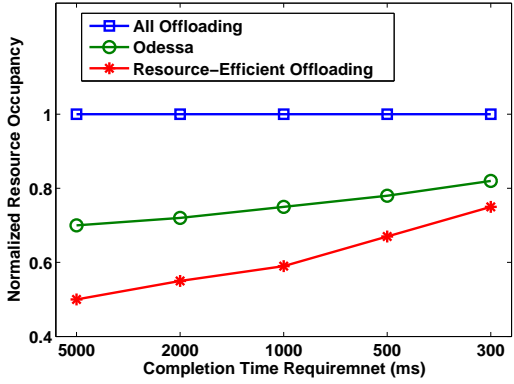}
\caption{\label{fig:Ocuupancy} Resource occupancy by different computation offloading schemes}
\end{figure}

\begin{figure}[tt]
\centering
\includegraphics[scale=0.45]{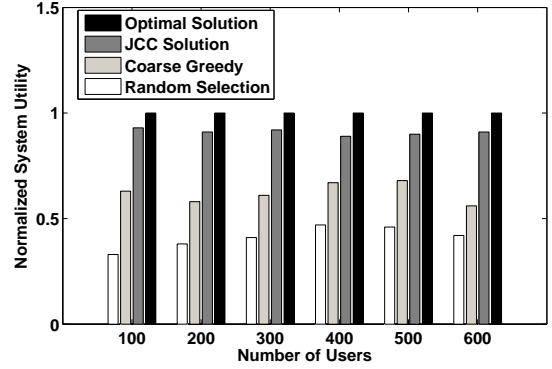}
\caption{\label{fig:UserCost} System performance with different number of users}
\end{figure}


\begin{figure}[tt]
\centering
\includegraphics[scale=0.45]{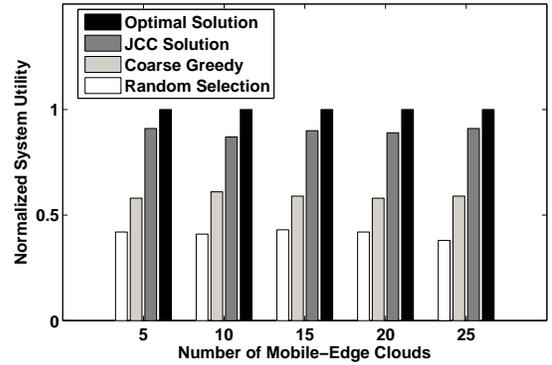}
\caption{\label{fig:Cloud}System performance with different number of mobile-edge clouds}
\end{figure}

\begin{figure}[tt]
\centering
\includegraphics[scale=0.45]{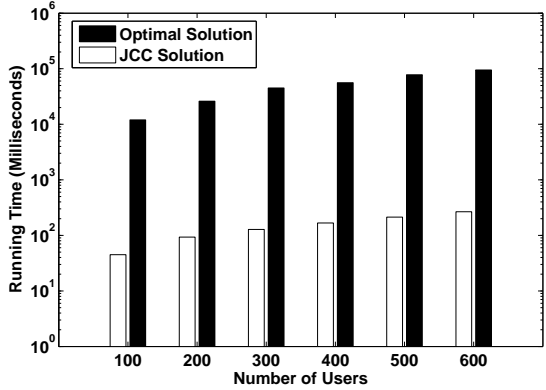}
\caption{\label{fig:Time}Running time of two comparing schemes}
\end{figure}



In this section, we evaluate the performance of the proposed on-demand JCC resource allocation mechanisms  through numerical studies.  We consider that there are $N=400$ randomly-scattered  mobile users requesting mobile-edge cloud service. There are multiple wireless base-stations, each of which has $M_{k}=15$ subchannels (each with a bandwidth $w=1$ Mhz) for supporting on-demand mobile-edge cloud service. The transmission power of a mobile user is $\eta_n=100$ mW and the background noise $\varpi_n=-100$ dBm \cite{saraydar2002efficient}. According to the physical interference model, we set the channel gain $h_{n,k(n)}=d^{-\alpha}_{n,k(n)}$ where $d_{n,k(n)}$ is the distance between user $n$ and its associated base-station $k(n)$ and $\alpha=4$ is the path loss factor \cite{saraydar2002efficient,li2018achievable}.

For the computation task, we consider two types of applications -- face recognition \cite{soyata2012cloud} and high-resolution QR-code recognition \cite{yang2013framework}, and use the task graphs therein. The task completion deadline of each user is randomly determined from the set $\{300, 500, 1000, 2000, 5000\}$ ms, and the computation capability of a user is randomly assigned from the set $\{0.5, 0.8, 1.0\}$ GHz \cite{wen2012energy}. For the cloud computing, there are multiple mobile-edge clouds, and each of which offers three types of virtual machines with the computation capabilities of $5, 10, 20$ GHz, respectively \cite{zhou2016energy,li2018dynamic}. The total computation capability of each mobile-edge cloud server is randomly assigned from the set $\{50, 100, 200\}$ GHz \cite{wen2012energy}. To reduce the connection latency, we assume that each base-station is connected to the closest mobile-edge cloud and each user is associated with the closest base-station. The truthful valuation $v_{n}$ of each user for purchasing the mobile-edge cloud service to complete the task is randomly assigned from the set $\{1, 2,...,20\}$ dollars.

We first evaluate the performance of our proposed resource-efficient offloading algorithm (i.e., delay-aware task graph partition in Algorithm \ref{alg:Delay-Optimal-Task-Graph} in terms of the resource occupancy. As the benchmark, we also implement the following computation offloading schemes: 1) all offloading solution -- all task components in the task graph will be offloaded; 2) Odessa scheme \cite{ra2011odessa} -- a task component is decided to be executed locally or offloaded remotely in a greedy manner. For each offloading scheme we find the minimum resource occupancy $\Phi_{n}(q_{n},s_{n})$ that satisfies the given completion time requirement. The results are shown in Figure \ref{fig:Ocuupancy} using face recognition \cite{soyata2012cloud} as a study case, wherein we use the all offloading solution as the baseline to normalize the resource occupancy performance by the other two schemes. We see that the proposed delay-aware task graph partition in Algorithm \ref{alg:Delay-Optimal-Task-Graph} can reduce up-to 50\% resource occupancy over the all offloading scheme. Compared to Odessa scheme, our algorithm can achieve up-to 22\% resource occupancy reduction. This demonstrates the proposed delay-aware task graph partition algorithm has superior resource efficiency.

We then evaluate the proposed approximate admission control in JCC resource allocation. As the benchmark, we consider the following solutions: (1) Optimal solution: we compute the (near) optimal solution of the admission control problem using Cross Entropy method, which is an advanced randomized searching technique and has been shown to be efficient in finding near-optimal solutions to complex combinatorial optimization problems \cite{rubinstein2004cross}; (2) Coarse greedy solution: instead of using the ranking criterion in (\ref{eq:pr}), we compute the greedy solution according to the coarse ranking criterion of $\lambda_{n}$, i.e., the claimed valuation in the bid; (3) Random selection: we sequentially and randomly select a new user to add into the accepted user set until the capacity constraints are violated.

In the following, we use the optimal solution by Cross Entropy method as the baseline to normalize the system utility (i.e., social welfare) of other solutions. We implement the simulations with different number of users, and mobile-edge clouds in Figures \ref{fig:UserCost} and \ref{fig:Cloud}, respectively.  In both cases, we observe that the proposed JCC solution achieves at least $34.7\%$ and $88.3\%$ performance improvement over the solutions of coarse greedy and random selection, respectively. Compared with the optimal solution, we see that the performance loss of the  JCC solution is at most $14.3\%$. This demonstrates of the efficiency of the proposed JCC solution.

We further investigate the running time of the proposed admission control algorithm.  We measure the running time in the computing environment of a 64-bit Windows PC with 2.5GHz Quad core CPU and 16GB memory. The results are shown in Figure \ref{fig:Time}. We see that the JCC mechanism is computationally efficient (with a running time of several milliseconds), with a factor of up to 1000 speed-up compared with the optimal solution by Cross Entropy method. This also demonstrates the proposed admission control algorithm is useful for practical implementation.

\begin{figure}[tt]
\centering
\includegraphics[scale=0.45]{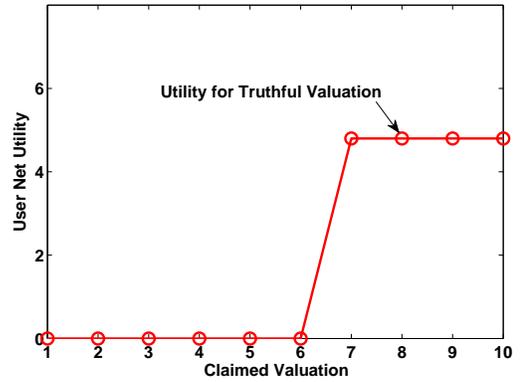}
\caption{\label{fig:Bid1}User net utility with different claimed valuations (truthful valuation $v_n$=8)}
\end{figure}

\begin{figure}[tt]
\centering
\includegraphics[scale=0.45]{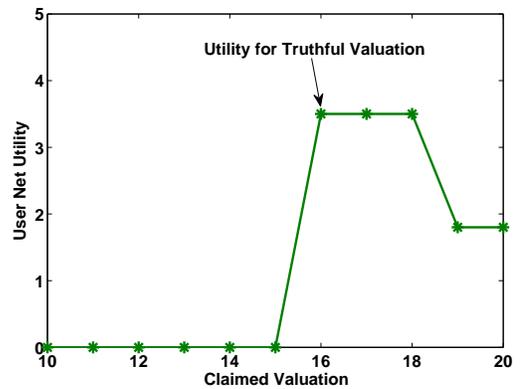}
\caption{\label{fig:Bid2}User net utility with different claimed valuations (truthful valuation $v_n$=16)}
\end{figure}

We next evaluate the truthfulness of the proposed pricing scheme for JCC resource allocation. We randomly select two users from the system and allow them to claim different valuations in their reports. We show the obtained user utility for these two users in Figures \ref{fig:Bid1} and \ref{fig:Bid2}, respectively. We see that by claimed valuations that are different from the truthful valuations, no users can improve their net utility (i.e., true valuation minus payment), under some cases the utility would be worsen. This demonstrates the proposed pricing scheme for JCC resource allocation is truthful such that no user has the incentive to lie about its valuation.

\section{Related Work}\label{sec:Related-Work}
Many previous works in mobile cloud computing have focused on the computation offloading problem (e.g., \cite{wen2012energy,yang2013framework,ChenToN2015,pu2016d2d,chen2018socially} and references therein).  Wen \emph{et al.} in \cite{wen2012energy} presented an efficient offloading policy by jointly configuring the clock frequency in the mobile device and scheduling the data transmission to minimize the energy consumption. Yang \emph{et al.} in \cite{yang2013framework} studied the scenario that multiple users share the wireless network bandwidth, and solved the problem of minimizing the delay performance by a centralized heuristic genetic algorithm. \rev{Chen \emph{et al.} in \cite{ChenToN2015,chen2015decentralized} proposed a game theoretic approach for designing decentralized multi-user computation offloading mechanism. Note that the computation offloading scheme in our paper is different from the studies above, since the existing works aim at optimizing user's energy efficiency or delay performance, while our objective is to optimize the user's offloading decision in terms of resource utilization efficiency. Moreover, existing works above do not consider the JCC resource allocation. Sardellitti \emph{et al.} in \cite{barbarossa2013joint} proposed an optimization framework for the energy-efficient JCC resource allocation, by assuming the network operator knows the complete information of all users' applications. Li \emph{et al.} in \cite{li2018learning} studied the computation offloading issue for deep learning applications. Tao \emph{et al.} in \cite{tao2017performance} investigated the performance guaranteed computation offloading for mobile-edge cloud computing. Chen \emph{et al.} in \cite{chen2017exploiting}  considered to exploit the D2D communications to assist the mobile edge computing. Along a different line, in this paper we introduce a novel paradigm of on-demand mobile-edge cloud service in order to efficiently allocate the JCC resources to those users value them most. Moreover, in our setting each user can locally decide the communication and computation resource profile for computation offloading without the need of reporting the application information to the network operator. This is very useful for reducing the information exchange overhead and protecting user's privacy without exposing the sensitive user specific application information.}

There are existing studies in designing truthful pricing scheme for  wireless resource auction (e.g., \cite{zhou2008ebay,gao2011map,dong2014double})  and cloud resource auction (e.g.,\cite{liu2017context,nisan2007algorithmic,zhang2013framework}), which consider either network communication or cloud computation resource allocation in a separate manner. While, in this paper we investigate the joint communication and computation resource allocation for mobile-edge cloud computing -- the nexus between cloud computing and wireless networking. Moreover, the truthful pricing scheme in our study builds upon the approximate admission control algorithm using the ranking metric in (\ref{eq:pr}) regarding to the resource occupancy function (\ref{ROF}), which is due to user's decision in the computation offloading. For the performance analysis we derive the approximation ratio for both the special and general cases. The analysis and the approximation ratios are completely different from the auction winner selection schemes in existing studies (e.g., \cite{nisan2007algorithmic,zhou2008ebay,gao2011map,dong2014double,zhang2013framework}).


\section{Conclusion}\label{sec:Conclusion}
Aiming at  provisioning flexible mobile-edge cloud service, in this paper we proposed a comprehensive framework consisting of a resource-efficient computation offloading mechanism for the users and a joint communication and computation (JCC) resource allocation mechanism for the network operator. We first solved the resource-efficient computation offloading problem for each individual user, and derive the optimal communication and computation resource demanding profile that minimizes the resource occupancy and meanwhile satisfies the delay constraint. We tackled the admission control problem for JCC resource allocation, and developed an efficient approximation solution of a low complexity. We also addressed the truthful pricing problem by resorting to the powerful tool of critical value approach. Extensive performance evaluation demonstrates that the proposed mechanisms can achieve superior performance for on-demand mobile-edge cloud computing.

For the future work, we are going to consider the more general case that users may dynamically depart and leave the mobile-edge cloud system. We will take into account users' mobility patterns and devise efficient online resource allocation algorithms to cope with such system dynamics.

\bibliographystyle{ieeetran}
\bibliography{MobileCloud}

\begin{thebibliography}{10}
\providecommand{\url}[1]{#1}
\csname url@samestyle\endcsname
\providecommand{\newblock}{\relax}
\providecommand{\bibinfo}[2]{#2}
\providecommand{\BIBentrySTDinterwordspacing}{\spaceskip=0pt\relax}
\providecommand{\BIBentryALTinterwordstretchfactor}{4}
\providecommand{\BIBentryALTinterwordspacing}{\spaceskip=\fontdimen2\font plus
\BIBentryALTinterwordstretchfactor\fontdimen3\font minus
  \fontdimen4\font\relax}
\providecommand{\BIBforeignlanguage}[2]{{%
\expandafter\ifx\csname l@#1\endcsname\relax
\typeout{** WARNING: IEEEtran.bst: No hyphenation pattern has been}%
\typeout{** loaded for the language `#1'. Using the pattern for}%
\typeout{** the default language instead.}%
\else
\language=\csname l@#1\endcsname
\fi
#2}}
\providecommand{\BIBdecl}{\relax}
\BIBdecl

\bibitem{kumar2010cloud}
K.~Kumar and Y.~Lu, ``Cloud computing for mobile users: Can offloading
  computation save energy?'' \emph{IEEE Computer}, vol.~43, no.~4, pp. 51--56,
  2010.

\bibitem{soyata2012cloud}
T.~Soyata, R.~Muraleedharan, C.~Funai, M.~Kwon, and W.~Heinzelman,
  ``Cloud-vision: Real-time face recognition using a mobile-cloudlet-cloud
  acceleration architecture,'' in \emph{IEEE Symposium on Computers and
  Communications (ISCC)}, 2012.

\bibitem{chen2015exploiting}
X.~Chen, B.~Proulx, X.~Gong, and J.~Zhang, ``Exploiting social ties for
  cooperative d2d communications: A mobile social networking case,''
  \emph{IEEE/ACM Transactions on Networking}, vol.~23, no.~5, pp. 1471--1484,
  2015.

\bibitem{MEC2014}
{European Telecommunications Standards Institute}, ``Mobile-edge computing --
  introductory technical white paper,'' September 2014.

\bibitem{wen2012energy}
Y.~Wen, W.~Zhang, and H.~Luo, ``Energy-optimal mobile application execution:
  Taming resource-poor mobile devices with cloud clones,'' in \emph{IEEE
  INFOCOM}, 2012.

\bibitem{yang2013framework}
L.~Yang, J.~Cao, Y.~Yuan, T.~Li, A.~Han, and A.~Chan, ``A framework for
  partitioning and execution of data stream applications in mobile cloud
  computing,'' \emph{ACM SIGMETRICS Performance Evaluation Review}, vol.~40,
  no.~4, pp. 23--32, 2013.

\bibitem{ChenToN2015}
X.~Chen, L.~Jiao, W.~Li, and X.~Fu, ``Efficient multi-user computation
  offloading for mobile-edge cloud computing,'' \emph{IEEE/ACM Transactions on
  Networking}, vol.~24, no.~5, pp. 2795--2808, 2016.

\bibitem{zhang2012offload}
Y.~Zhang, H.~Liu, L.~Jiao, and X.~Fu, ``To offload or not to offload: an
  efficient code partition algorithm for mobile cloud computing,'' in
  \emph{IEEE 1st International Conference on Cloud Networking (CLOUDNET)},
  2012.

\bibitem{kao2014optimizing}
Y.-H. Kao and B.~Krishnamachari, ``Optimizing mobile computational offloading
  with delay constraints,'' in \emph{IEEE GLOBECOM}.\hskip 1em plus 0.5em minus
  0.4em\relax IEEE, 2014, pp. 2289--2294.

\bibitem{barbarossa2013joint}
S.~Sardellitti, G.~Scutari, and S.~Barbarossa, ``Joint optimization of radio
  and computational resources for multicell mobile-edge computing,'' in
  \emph{IEEE Transactions on Signal and Information Processing over
  Networks}.\hskip 1em plus 0.5em minus 0.4em\relax IEEE, 2015.

\bibitem{chen2018thriftyedge}
X.~Chen, Q.~Shi, L.~Yang, and J.~Xu, ``Thriftyedge: Resource-efficient edge
  computing for intelligent iot applications,'' \emph{IEEE Network}, vol.~32,
  no.~1, pp. 61--65, 2018.

\bibitem{saraydar2002efficient}
C.~U. Saraydar, N.~B. Mandayam, and D.~J. Goodman, ``Efficient power control
  via pricing in wireless data networks,'' \emph{IEEE Transactions on
  Communications}, vol.~50, no.~2, pp. 291--303, 2002.

\bibitem{cormen2009introduction}
C.~Stein, T.~Cormen, R.~Rivest, and C.~Leiserson, \emph{Introduction to
  algorithms}.\hskip 1em plus 0.5em minus 0.4em\relax MIT press, 2009.

\bibitem{caprara2004two}
A.~Caprara and M.~Monaci, ``On the two-dimensional knapsack problem,''
  \emph{Operations Research Letters}, vol.~32, no.~1, pp. 5--14, 2004.

\bibitem{milgrom2004putting}
P.~R. Milgrom, \emph{Putting auction theory to work}.\hskip 1em plus 0.5em
  minus 0.4em\relax Cambridge University Press, 2004.

\bibitem{li2018achievable}
Z.~Li, F.~Xiao, S.~Wang, T.~Pei, and J.~Li, ``Achievable rate maximization for
  cognitive hybrid satellite-terrestrial networks with af-relays,'' \emph{IEEE
  Journal on Selected Areas in Communications}, vol.~36, no.~2, pp. 304--313,
  2018.

\bibitem{zhou2016energy}
Z.~Zhou, M.~Dong, K.~Ota, G.~Wang, and L.~T. Yang, ``Energy-efficient resource
  allocation for d2d communications underlaying cloud-ran-based lte-a
  networks,'' \emph{IEEE Internet of Things Journal}, vol.~3, no.~3, pp.
  428--438, 2016.

\bibitem{li2018dynamic}
Z.~Li, B.~Chang, S.~Wang, A.~Liu, F.~Zeng, and G.~Luo, ``Dynamic compressive
  wide-band spectrum sensing based on channel energy reconstruction in
  cognitive internet of things,'' \emph{IEEE Transactions on Industrial
  Informatics}, vol.~14, no.~6, pp. 2598 -- 2607, 2018.

\bibitem{ra2011odessa}
M.-R. Ra, A.~Sheth, L.~Mummert, P.~Pillai, D.~Wetherall, and R.~Govindan,
  ``Odessa: enabling interactive perception applications on mobile devices,''
  in \emph{Proceedings of the 9th international conference on Mobile systems,
  applications, and services}.\hskip 1em plus 0.5em minus 0.4em\relax ACM,
  2011, pp. 43--56.

\bibitem{rubinstein2004cross}
R.~Y. Rubinstein and D.~P. Kroese, \emph{The cross-entropy method: a unified
  approach to combinatorial optimization, Monte-Carlo simulation and machine
  learning}.\hskip 1em plus 0.5em minus 0.4em\relax Springer, 2004.

\bibitem{pu2016d2d}
L.~Pu, X.~Chen, J.~Xu, and X.~Fu, ``D2d fogging: An energy-efficient and
  incentive-aware task offloading framework via network-assisted d2d
  collaboration,'' \emph{IEEE Journal on Selected Areas in Communications},
  vol.~34, no.~12, pp. 3887--3901, 2016.

\bibitem{chen2018socially}
X.~Chen, Z.~Zhou, W.~Wu, D.~Wu, and J.~Zhang, ``Socially-motivated cooperative
  mobile edge computing,'' \emph{IEEE Network}, 2018.

\bibitem{chen2015decentralized}
X.~Chen, ``Decentralized computation offloading game for mobile cloud
  computing,'' \emph{IEEE Transactions on Parallel and Distributed Systems},
  vol.~26, no.~4, pp. 974--983, 2015.

\bibitem{li2018learning}
H.~Li, K.~Ota, and M.~Dong, ``Learning iot in edge: Deep learning for the
  internet of things with edge computing,'' \emph{IEEE Network}, vol.~32,
  no.~1, pp. 96--101, 2018.

\bibitem{tao2017performance}
X.~Tao, K.~Ota, M.~Dong, H.~Qi, and K.~Li, ``Performance guaranteed computation
  offloading for mobile-edge cloud computing,'' \emph{IEEE Wireless
  Communications Letters}, 2017.

\bibitem{chen2017exploiting}
X.~Chen, L.~Pu, L.~Gao, W.~Wu, and D.~Wu, ``Exploiting massive d2d
  collaboration for energy-efficient mobile edge computing,'' \emph{IEEE
  Wireless Communications}, vol.~24, no.~4, pp. 64--71, 2017.

\bibitem{zhou2008ebay}
X.~Zhou, S.~Gandhi, S.~Suri, and H.~Zheng, ``ebay in the sky: Strategy-proof
  wireless spectrum auctions,'' in \emph{Proceedings of the 14th ACM
  international conference on Mobile computing and networking}.\hskip 1em plus
  0.5em minus 0.4em\relax ACM, 2008, pp. 2--13.

\bibitem{gao2011map}
L.~Gao, Y.~Xu, and X.~Wang, ``Map: Multiauctioneer progressive auction for
  dynamic spectrum access,'' \emph{IEEE Transactions on Mobile Computing},
  vol.~10, no.~8, pp. 1144--1161, 2011.

\bibitem{dong2014double}
W.~Dong, S.~Rallapalli, L.~Qiu, K.~Ramakrishnan, and Y.~Zhang, ``Double
  auctions for dynamic spectrum allocation,'' in \emph{IEEE INFOCOM}.\hskip 1em
  plus 0.5em minus 0.4em\relax IEEE, 2014, pp. 709--717.

\bibitem{liu2017context}
Y.~Liu, A.~Liu, S.~Guo, Z.~Li, Y.-J. Choi, and H.~Sekiya, ``Context-aware
  collect data with energy efficient in cyber--physical cloud systems,''
  \emph{Future Generation Computer Systems}, 2017.

\bibitem{nisan2007algorithmic}
N.~Nisan, T.~Roughgarden, E.~Tardos, and V.~V. Vazirani, \emph{Algorithmic game
  theory}.\hskip 1em plus 0.5em minus 0.4em\relax Cambridge University Press
  Cambridge, 2007, vol.~1.

\bibitem{zhang2013framework}
H.~Zhang, B.~Li, H.~Jiang, F.~Liu, A.~V. Vasilakos, and J.~Liu, ``A framework
  for truthful online auctions in cloud computing with heterogeneous user
  demands,'' in \emph{Proceedings IEEE INFOCOM}, 2013.

\end{thebibliography}

\end{document}